\documentclass[12pt,a4paper]{article}

\usepackage{amsmath,amsthm,amssymb,amscd,a4wide}

\usepackage{dsfont}
\usepackage{mathrsfs}
\usepackage{setspace}

\newcommand{\ud}{\mathrm{d}}

\newcommand{\cD}{{\mathcal D}}
\newcommand{\cE}{{\mathcal E}}
\newcommand{\cH}{{\mathcal H}}
\newcommand{\cL}{{\mathcal L}}
\newcommand{\cV}{{\mathcal V}}
\newcommand{\vl}{\boldsymbol{l}}
\newcommand{\vy}{\boldsymbol{y}}

\newcommand{\rz}{{\mathbb R}}

\newcommand{\kz}{{\mathbb C}}

\newcommand{\M}{\mathrm{M}}

\newcommand{\eins}{\mathds{1}}

\DeclareMathOperator{\diag}{diag}
\DeclareMathOperator{\op}{op}
\DeclareMathOperator{\ran}{ran}

\numberwithin{equation}{section}

\newtheorem{theorem}{Theorem}[section]

\newtheorem{prop}[theorem]{Proposition}

\theoremstyle{definition}
\newtheorem{defn}[theorem]{Definition}

\begin{document}

\thispagestyle{empty}

\vspace*{1cm}

\begin{center}

{\LARGE\bf Quantum graphs with two-particle\\[5mm]
 contact interactions} \\

\vspace*{3cm}

{\large Jens Bolte}%
\footnote{E-mail address: {\tt jens.bolte@rhul.ac.uk}}
{\large and Joachim Kerner}%
\footnote{E-mail address: {\tt joachim.kerner.2010@live.rhul.ac.uk}}
\vspace*{2cm}

Department of Mathematics\\
Royal Holloway, University of London\\
Egham, TW20 0EX\\
United Kingdom\\

\end{center}

\vfill

\begin{abstract}
We construct models of many-particle quantum graphs with singular two-particle 
contact interactions, which can be either hardcore- or $\delta$-interactions.
Self-adjoint realisations of the two-particle Laplacian including such
interactions are obtained via their associated quadratic forms. We
prove discreteness of spectra as well as Weyl laws for the asymptotic
eigenvalue counts. These constructions are first performed for two 
distinguishable particles and then for two identicle bosons. Furthermore, we 
extend the models to $N$ bosons with two-particle interactions, thus 
implementing the Lieb-Liniger model on a graph.
\end{abstract}

\newpage

\section{Introduction}
This paper is the second in a serious of papers developing models of 
interacting, non-relativisitc many-particle systems on (compact) graphs. 
The underlying one-particle quantum graphs are well established models for a 
variety of problems in quantum mechanics. Interest in such one-particle models 
surged after Kottos and Smilansky \cite{KotSmi99} had discovered that the 
eigenvalues of quantum Hamiltonians describing single particles on graphs 
possess the same correlations as eigenvalues of random hermitian matrices. 
Henceforth, quantum graphs proved to be very successful models in the area of 
quantum chaos \cite{GnuSmi07} and beyond \cite{Exnetal08}.

In our first paper \cite{BolKer11} we introduced systems of two particles on 
compact metric graphs. We identified self-adjoint realisations of the 
two-particle Laplacian that describe singular two-particle interactions
that are located in the vertices of the graph. In that context we analysed the 
Laplacians indirectly, by first constructing closed, semi-bounded quadratic 
forms and then identifying the self-adjoint operators that are uniquely 
associated with the quadratic forms. Using quadratic forms allowed us, 
furthermore, to prove that the Laplacians have compact resolvents and thus 
possess purely discrete spectra. Moreover, with a bracketing argument we were 
able to prove a Weyl law for the asymptotic eigenvalue count.

The goal of this paper now is to introduce two-particle interactions of a 
different kind, modelling short-range interactions in terms of singular contact 
interactions that are formally given by a Hamiltonian 
\begin{equation}
\label{HamiltonianDelta}
 H = -\Bigl(\frac{\partial^{2}}{\partial x^{2}}+
     \frac{\partial^{2}}{\partial y^{2}}\Bigr) + \alpha\,\delta(x-y)\ .
\end{equation}
Here $x$ and $y$ are coordinates for the positions of the two particles on the
edges of the graph and $\alpha\in\rz$ is a coupling parameter, such that 
$\alpha>0$ corresponds to repulsive interactions. Contact interactions of this 
kind play a prominent role in the description of Bose-Einstein condensates 
(BEC), leading to Gross-Pitaevskii equations, in plasma physics and in solid 
state physics (see, e.g., \cite{Lieetal05,Cazetal11}). 
They are also of interest from a 
mathematical point of view: models of (an arbitrary number of) particles on 
the real line with Hamiltonian \eqref{HamiltonianDelta} are completely solvable
in the sense that the eigenfunctions can be constructed explicitly, see 
\cite{Ya67,AbFeKu02,AbFeKu04}. As in the one-particle case the connectivity 
of a graph, however, adds sufficient complexity to the problem to render it 
not solvable in that sense. Therefore, bosonic many-particle systems on graphs 
are expected to show generic many-particle properties of systems confined to 
one spatial dimension.

In this paper we shall follow our previous approach \cite{BolKer11} in that 
we first introduce suitable quadratic forms, show that they are closed and 
semi-bounded, and then identify the corresponding self-adjoint operators. 
These operators shall always be two-particle Lapacians, however, with 
suitable domains that implement singular two-particle contact interactions. 
The resulting operators are then rigorous versions of the formal
Hamiltonian~\eqref{HamiltonianDelta}. In order to achieve this the domains
of the operators are required to contain jump conditions along diagonals $x=y$ 
for derivatives of the functions in the domain. A similar construction was 
already given by Harmer \cite{Har07a,Har08}, who considered two particles on 
a star graph (with semi-infinite edges). Our construction works for any compact
metric graph and allows for straight-forward generalisations to somewhat 
broader classes of singular contact interactions of which the $\delta$-type 
interactions \eqref{HamiltonianDelta} form a prominent subclass.

The paper is organised as follows: In Section~\ref{1sec} we review important
properties of one-particle quantum graphs and introduce relevant concepts 
and notations that we are using in the following sections.  Section~\ref{2sec} 
is then devoted to the construction of the contact interactions via quadratic 
forms for two distinguishable particles on a general compact, metric graph.
In that context we encounter the problem of  elliptic regularity \cite{Dob05} 
in the same way as previously \cite{BolKer11}, suggesting an analogous solution
for a similar class of boundary conditions.  
In Section~\ref{3sec} we first implement a bosonic realisation of a particle 
exchange symmetry and hence obtain a rigorous construction of the
bosonic version of the operator \eqref{HamiltonianDelta}. We then extend 
this construction to $N$ bosons on a general compact graph, for which the 
equivalent to the formal operator \eqref{HamiltonianDelta} is
\begin{equation}
\label{HamiltonianDeltaN}
 H = -\Delta_N + \alpha\sum_{i<j}^N\delta(x_i-x_j)\ .
\end{equation}
Here $x_1,\dots,x_N$ are the particle positions and $-\Delta_N$ is the
Laplacian in these $N$ variables. Hence, the model we are investigating
is an extension to quantum graphs of the well-known Lieb-Liniger model 
\cite{LieLin63} that has been studied extensively in the context of BEC.

\section{Preliminaries}
\label{1sec}
Before introducing many-particle systems on graphs we briefly describe 
one-particle quantum graphs. They form the basis for the tensor-product 
construction of many-particle quantum systems on graphs.

The classical configuration space of a quantum graph is a compact 
metric graph, i.e., a finite graph $\Gamma = (\cV,\cE)$ with vertices 
$\cV = \{v_1,\dots,v_V\}$ and edges $\cE=\{e_1,\dots,e_E\}$. The latter 
are identified with intervals $[0,l_e]$, $e=1,\dots,E$, thus introducing 
a metric on the graph, see \cite{KotSmi99,KosSch99,Kuc04,GnuSmi07} for
details.

Functions $F$ on the graph are collections of functions on the edges,
\begin{equation}
 F=(f_1,\dots,f_E) \ ,\quad\text{with}\quad f_e : [0,l_e]\to\kz\ ,
\end{equation}
so that spaces of functions on $\Gamma$ are (finite) direct sums of
the respective spaces of functions on the edges. The most relevant
space is the one-particle Hilbert space
\begin{equation}
 \cH_1 = L^2 (\Gamma) := \bigoplus_{e=1}^E L^2 (0,l_e) \ ,
\end{equation}
and all other spaces are constructed in a similar way.

The one-particle Hamiltonian is a Laplacian, acting as a second-order 
differentiation,
\begin{equation}
 -\Delta_1 F= (-f_1'',\dots,-f_E'') \ ,
\end{equation}
on $F\in C^\infty (\Gamma)$. We here use the index to indicate that this 
is a one-particle Laplacian. 

Domains of self-adjoint realisations of the Laplacian are characterised in terms
of boundary conditions in the vertices. These require boundary values
\begin{equation}
\label{Fbv}
 F_{bv} := \bigl( f_1(0),\dots,f_E (0),f_1(l_1),\dots,f_E (l_E) \bigr)^T 
              \in\kz^{2E}\ ,
\end{equation}
of functions and (inward) derivatives,
\begin{equation}
 F_{bv}' := \bigl( f_1'(0),\dots,f_E' (0),-f_1'(l_1),\dots,-f_E' (l_E) \bigr)^T 
               \in\kz^{2E}\ ,
\end{equation}
as well as projectors $P$ and $Q=\eins_{2E} -P$ acting on the space $\kz^{2E}$ 
of boundary values, and self-adjoint endomorphisms $L$ of 
$\ran Q\subset \kz^{2E}$.

The self-adjoint realisations of $-\Delta_1$ can be identified via the quadratic
forms that are uniquely associated with them \cite{Kuc04}.
\begin{theorem}[Kuchment]
\label{KuchmentThm}
The quadratic form
\begin{equation}
\label{Qform1}
 Q^{(1)}_{P,L}[F] 
  = \int_\Gamma |\nabla f| \ \ud x -
        \langle F_{bv},LF_{bv}\rangle_{\kz^{2E}}   
  = \sum_{e=1}^E \int_0^{l_e} |f'_e (x)|^2 \ \ud x -
          \langle F_{bv},LF_{bv}\rangle_{\kz^{2E}}\ ,
\end{equation}
with domain
\begin{equation}
\label{Qformdomain}
 \cD_{Q^{(1)}} = \{F\in H^1(\Gamma);\ PF_{bv}=0\}
\end{equation}
is closed and bounded from below. The unique, self-adjoint and semi-bounded
operator associated with this form is the one-particle Laplacian $-\Delta_1$
with domain
\begin{equation}
\label{1partBCalt}
 \cD_1 (P,L) = \{ F\in H^2 (\Gamma);\  PF_{bv}=0\ \text{and}\ 
                  QF_{bv}'+LQF_{bv}=0 \}  \ .
\end{equation}
\end{theorem}

A two-particle quantum system requires the tensor product of two one-particle 
Hilbert spaces,
\begin{equation}
 \cH_2 := \cH_1 \otimes \cH_1 \ .
\end{equation}
For a quantum graph this means that
\begin{equation}
\label{2Hilbert}
 \cH_2 = \Bigl(\bigoplus_{e=1}^E L^2 (0,l_e)\Bigr) \otimes 
 \Bigl(\bigoplus_{e=1}^E L^2 (0,l_e)\Bigr) \ ,
\end{equation}
such that vectors $\Psi\in\cH_2$ are collections $\Psi = (\psi_{e_1 e_2})$ 
of $E^2$ functions defined on the rectangles 
$D_{e_1e_2}=(0,l_{e_1})\times(0,l_{e_2})$. Their disjoint union is
denoted as
\begin{equation}
\label{D_Gamdef}
 D_\Gamma = \dot{\bigcup_{e_1e_2}}D_{e_1e_2} \ ,
\end{equation}
so that one may view $\cH_2$ as
\begin{equation}
 L^2(D_\Gamma) := \bigoplus_{e_1e_2}L^2(D_{e_1e_2})  \ .
\end{equation}
We shall use a similar notation for other function spaces.

As a differential operator, the two-particle Laplacian acts as
\begin{equation}
\label{2Laplace}
 (-\Delta_{2}\Psi)_{e_1e_2} = -\frac{\partial^2\psi_{e_1e_2}}{\partial x_{e_1}^2}-
  \frac{\partial^2\psi_{e_1e_2}}{\partial x_{e_2}^2}\ ,
\end{equation}
and hence has the same form as a Laplacian in $\rz^2$. Defined on 
the domain $C^\infty_0 (D_\Gamma)$, this operator is  symmetric, but not 
self-adjoint.

Self-adjoint realisations of the two-particle Laplacians can either represent
non-inter\-acting particles, or introduce two-particle interactions via
boundary conditions. A particular class of singular two-particle interactions
that are localised in the vertices was established in \cite{BolKer11}. Here
we shall introduce two-particle contact interactions that are localised on the 
edges.

\section{Contact interactions on a general compact graph}
\label{2sec}
The interactions we have in mind are intended to model two point-like particles 
on a graph that interact when they hit each other, i.e., when they are
located in the same position. This requires, in particular,
that they are on the same edge. Hence, the subset of the two-particle
configuration space $D_\Gamma$ \eqref{D_Gamdef} where these
interactions take place consists of the diagonals of the squares $D_{ee}$. 
Singular interactions require a dissection of the configuration space
along these subspaces, and suitable matching conditions for functions
and their derivatives along the boundaries introduced by the dissection. 
We therefore define the `dissected' configuration space
\begin{equation}
\label{Ddecomp}
 D^\ast_\Gamma := \left(\dot{\bigcup_{e_1\neq e_2}}D_{e_1e_2}\right) 
 \dot{\bigcup_{e}}\left(D_{ee}^+\dot\cup D_{ee}^-\right)\ ,
\end{equation}
where $D_{ee}^+=\{(x,y)\in D_{ee};\ x>y\}$ and 
$D_{ee}^-=\{(x,y)\in D_{ee};\ x<y\}$. Functions on $D^\ast_\Gamma$ are denoted
as $\Psi=(\psi_{e_1e_2})$. The components $\psi_{e_1e_2}$ for $e_1\neq e_2$ are 
defined on $D_{e_1e_2}$, whereas $\psi_{ee}=(\psi_{ee}^+,\psi_{ee}^-)$ with
$\psi_{ee}^\pm$ defined on $D^\pm_{ee}$.

The two-particle Hilbert space $\cH_2$ \eqref{2Hilbert} can then also be 
viewed as
\begin{equation}
 L^2(D^\ast_\Gamma) = \left(\bigoplus_{e_1\neq e_2}L^2(D_{e_1e_2}) 
 \right)\bigoplus_{e}\bigl(L^2(D_{ee}^+)\oplus L^2(D_{ee}^-)\bigr)\ .
\end{equation}
Boundary values of functions $\Psi\in H^1(D^\ast_\Gamma)$ are encoded 
in vectors 
\begin{equation}
\label{bvdef}
 \Psi_{bv}(y) = \bigl(\psi_{e_1e_2,bv}(y)\bigr) \qquad\text{and}\qquad
  \Psi'_{bv}(y) = \bigl(\psi'_{e_1e_2,bv}(y)\bigr)\ .
\end{equation}
We distinguish components with $e_1\neq e_2$ from those with $e_1=e_2$,
as in the latter case additional boundary values along diagonals have to be
taken into account. More specifically, when $e_1\neq e_2$ we define
\begin{equation}
\label{bvnondiag}
 \psi_{e_1e_2,bv}(y) := 
 \begin{pmatrix}\sqrt{l_{e_2}}\psi_{e_1 e_2}(0,l_{e_2}y) \\ 
 \sqrt{l_{e_2}}\psi_{e_1 e_2}(l_{e_1},l_{e_2}y)\\
 \sqrt{l_{e_1}}\psi_{e_1 e_2}(l_{e_1}y,0) \\ 
 \sqrt{l_{e_1}}\psi_{e_1 e_2}(l_{e_1}y,l_{e_2})
 \end{pmatrix} 
 \qquad\text{and}\qquad 
 \psi'_{e_1e_2,bv}(y) := \begin{pmatrix}\sqrt{l_{e_2}}\psi_{e_1 e_2,x}(0,l_{e_2}y) \\ 
 -\sqrt{l_{e_2}}\psi_{e_1 e_2,x}(l_{e_1},l_{e_2}y) \\
 \sqrt{l_{e_1}}\psi_{e_1 e_2,y}(l_{e_1}y,0) \\ 
 -\sqrt{l_{e_1}}\psi_{e_1 e_2,y}(l_{e_1}y,l_{e_2})
 \end{pmatrix} \ ,
\end{equation}
where $y\in [0,1]$. 
When $e_1=e_2$ boundary values along the diagonals of the squares
$D_{ee}$ have to be added, including those for derivatives. Noting that the 
inward normal derivatives along the `diagonal' part of the boundary of 
$D^\pm_{ee}$ are
\begin{equation}
\label{nderdiag}
 \psi^\pm_{ee,n} = 
 \frac{\pm 1}{\sqrt{2}}\bigl(\psi^\pm_{ee,x}-\psi^\pm_{ee,y}\bigr) \ ,
\end{equation}
we set
\begin{equation}
\label{bvdiag}
 \psi_{ee,bv}(y) := \begin{pmatrix}  \sqrt{l_e}\psi^{-}_{ee}(0,l_e y) \\ 
 \sqrt{l_e}\psi^{+}_{ee}(l_e,l_e y) \\ \sqrt{l_e}\psi^{+}_{ee}(l_e y,0) \\ 
 \sqrt{l_e}\psi^{-}_{ee}(l_e y,l_e) \\ \sqrt{l_e} \psi^{+}_{ee}(l_e y,l_e y) \\ 
 \sqrt{l_e}\psi^{-}_{ee}(l_e y,l_e y) \end{pmatrix}
 \qquad\text{and}\qquad
 \psi'_{ee,bv}(y) := \begin{pmatrix} \sqrt{l_{e}}\psi^{-}_{ee,x}(0,l_{e}y) \\ 
 -\sqrt{l_e}\psi^{+}_{ee,x}(l_e,l_e y) \\ \sqrt{l_e}\psi^{+}_{ee,y}(l_e y,0) \\ 
 -\sqrt{l_e}\psi^{-}_{ee,y}(l_e y,l_e) \\ \sqrt{2l_e}\psi^{+}_{ee,n}(l_e y,l_e y)\\ 
 \sqrt{2l_e}\psi^{-}_{ee,n}(l_e y,l_e y) \end{pmatrix}\ ,
\end{equation}
for $y \in [0,1]$. Altogether, the vectors \eqref{bvdef} of boundary values 
have $n(E):=4E^2 +2E$ components.

As a next step we introduce the bounded and measurable maps 
$P,L: [0,1] \to \M(n(E),\kz)$ that are required to fulfil
\begin{enumerate}
\item $P(y)$ is an orthogonal projector,
\item $L(y)$ is a self-adjoint endomorphism on $\ker P(y)$,
\end{enumerate}
for a.e. $y \in [0,1]$. We then introduce two bounded and self-adjoint 
operators, $\Pi$ and $\Lambda$, on $L^2(0,1)\otimes\kz^{n(E)}$. They are 
defined to act as $(\Pi\chi)(y):=P(y)\chi(y)$ and  
$(\Lambda\chi)(y):=L(y)\chi(y)$ on $\chi\in L^2(0,1)\otimes\kz^{n(E)}$.

Our aim is to obtain self-adjoint realisations of the two-particle Laplacian 
$-\Delta_2$, see \eqref{2Laplace}, that represent two-particle contact 
interactions and are extensions of $-\Delta_{2,0}$ defined on the domain
$C^\infty_0(D^\ast_\Gamma)$. In analogy to the one-particle case 
\eqref{1partBCalt}, as well as the case of singular interactions covered in 
\cite{BolKer11}, their domains should be given in the form
\begin{equation}
\label{DomainGen}
\begin{split}
 \cD_2 (P,L) := \{
       &\Psi\in H^2(D^\ast_\Gamma);\ P(y)\Psi_{bv}(y)=0\ \text{and}\\
       &\quad Q(y)\Psi'_{bv}(y)+L(y)Q(y)\Psi_{bv}(y)=0\ \text{for a.e.}\
          y\in [0,l] \} \ .
\end{split}
\end{equation}
In order to clearly distinguish the boundary conditions that induce contact
interactions from other kinds of boundary conditions we rearrange the
order of terms in the boundary vectors. We first list, for each edge $e$, the 
lower two boundary values in \eqref{bvdiag}, and then, for each pair 
$(e_1,e_2)$, either the four components in \eqref{bvnondiag} or the upper four 
components of \eqref{bvdiag}. That way one achieves a decomposition of the
space of boundary values according to
\begin{equation}
\label{bvblocks}
 \kz^{n(E)} = V_{contact}\oplus V_{vertex}\ .
\end{equation}
Here $V_{contact}$, with $\dim V_{contact}=2E$, contains the boundary values
\eqref{bvdiag} along diagonals, whereas $V_{vertex}$, with $\dim V_{vertex}=4E^2$,
contains the remaining boundary values \eqref{bvnondiag} and \eqref{bvdiag}, 
which are associated with vertices. A separation of contact interactions from 
any other boundary effects requires to choose $P$ and $L$ as block-diagonal 
with respect to the decomposition \eqref{bvblocks}. From now on we assume this 
to be the case.

For the restriction of $P$ and $L$ to $V_{vertex}$ we assume the same conditions
as in \cite{BolKer11}. For most purposes, however, it is sufficient to suppose 
that there are no two-particle interactions in the vertices. In \cite{BolKer11}
the non-interacting boundary conditions were characterised as follows: The 
restrictions of $P$ and $L$ to $V_{vertex}$ are independent of $y$ and 
block-diagonal with respect to a decomposition of $V_{vertex}$ according to the 
index $e_2$ in \eqref{bvdef}.

The restriction of $P$ and $L$ to $V_{contact}$ should, first of all, be 
block-diagonal with respect to a decomposition of that space according to the 
edges in order to avoid `contact' interactions across edges. Further 
restriction are not necessary, but we identify the following two cases as
of particular interest because they correspond to a Hamiltonian of the form 
\eqref{HamiltonianDelta}. 
\begin{defn}
Let $\alpha:[0,1]\to\rz$ be Lipschitz continuous. Then a contact interaction 
is said to be of 
\begin{itemize}
\item[(i)] {\it $\delta$-type with (variable) strength $\alpha$}, if 
$\Psi\in H^2(D^\ast_\Gamma)$ is continuous across diagonals,
\begin{equation}
\label{fctcont}
 \psi^{+}_{ee}(l_e y,l_e y) =  \psi^{-}_{ee}(l_e y,l_e y) \ ,
\end{equation}
and satisfies jump conditions for the normal derivatives,
\begin{equation}
\label{derjump}
 \psi^+_{ee,n}(l_e y,l_e y) + \psi^-_{ee,n}(l_e y,l_e y) = 
 \frac{1}{\sqrt{2}}\alpha(y)\psi^\pm_{ee}(l_e y,l_e y) \ ,
\end{equation}
\item[(ii)] {\it hardcore type}, if it satisfies Dirichlet boundary
conditions along diagonals.
\end{itemize}
\end{defn}
We remark that contact interactions of the $\delta$-type can be seen as a
rigorous realisation of a Hamiltonian
\begin{equation}
 -\Delta_2 + \alpha(y)\,\delta(x-y)\ .
\end{equation}
The case $\alpha(y)>0$ for all $y\in[0,1]$ corresponds to repulsive 
interactions and is the most relevant case for models of actual particles on 
a graph. Hardcore interactions follow from such a formal Hamiltonian in the 
limit $\alpha\to\infty$.

Following our intention to represent domains of two-particle Laplacians in
the form \eqref{DomainGen} we have to choose the maps $P$ and $L$ in such a
way that their restrictions to the edge-$e$ subspace of $V_{contact}$ are
\begin{equation}
\label{BCdelta1}
 P_{contact,e}(y) = \frac{1}{2}\begin{pmatrix} 1 & -1 \\ -1 & 1\end{pmatrix} \,
\end{equation}
and
\begin{equation}
\label{BCdelta2}
 L_{contact,e}(y) = -\frac{1}{2}\,\alpha(y)\,\eins_2 \ ,
\end{equation}
in order to generate $\delta$-type contact interactions. Hardcore interactions
require the choice $P_{contact}=\eins$ and $L_{contact}=0$.

Our approach to self-adjoint realisations of the Laplacian uses suitable 
quadratic forms, which are uniquely associated with these operators. 
\begin{prop}
\label{2quadformDelta}
Assume that the maps $P,L:[0,1]\to\M(n(E),\kz)$ are bounded and measurable. 
Then the quadratic form
\begin{equation}
\label{QformDeltaGeneral}
 Q^{(2)}_{P,L}[\psi] = \langle\nabla\Psi,\nabla\Psi\rangle_{L^2(D^\ast_\Gamma)}
 -\int_0^1 \langle \Psi_{bv}(y),L(y)\Psi_{bv}(y)\rangle_{\kz^{n(E)}} \ud y \ ,
\end{equation}
with domain
\begin{equation}
\label{DefquadDelta}
 \cD_{Q^{(2)}} = \{\Psi\in H^1(D^{*}_{\Gamma});\ P(y)\Psi_{bv}(y)=0\ \text{for a.e.}\
 y \in [0,1]\}
\end{equation}
is closed and semi-bounded.
\end{prop}
\begin{proof}
The proof follows by using the same steps as in the corresponding proof in 
\cite{BolKer11}. The only consideration that has to be added concerns the 
upper bound
\begin{equation}
 \left| \int_0^1 \langle\Psi_{bv}(y),L(y)\Psi_{bv}(y) \rangle_{\kz^{n(E)}}\
 \ud y \right| \leq  L_{max}\, \|\Psi_{bv}\|^{2}_{L^{2}(0,1)\otimes\kz^{n(E)}} \ .
\end{equation}
To estimate the right-hand side one requires the bound,
\begin{equation}
\label{normeq}
 \|\Psi_{bv}\|^{2}_{L^2(0,1)\otimes\kz^{n(E)}} \leq K\,\left( \frac{2}{\delta}\,
 \|\Psi\|^{2}_{L^{2}(D^{*}_{\Gamma})} + \delta\, \|\nabla\Psi\|^{2}_{L^{2}(D^{*}_{\Gamma})}
 \right) \ ,
\end{equation}
to hold for all $\delta\leq\delta_0$, where $K,\delta_0>0$. The 
contribution from the rectangles $D_{e_1e_2}$ (with $e_1\neq e_2$) in the 
decomposition \eqref{Ddecomp} can be dealt with as in \cite{BolKer11} and
is based on a result in \cite{Kuc04}. For the triangles $D_{ee}^\pm$ we note 
that close to the corners with angles $\pi/4$ this method fails. However,
one can always reflect functions $\psi^\pm_{ee}$ across edges, define
them on suitable squares and then apply the bound as before for the 
rectangles. The proof then continues as in \cite{BolKer11}.  
\end{proof}
According to the representation theorem for quadratic forms (see, e.g.,
\cite{Kat66}) there exists
a unique self-adjoint and semibounded operator $H$ with domain 
$\cD(H)\subseteq\cD_{Q^{(2)}}$ that is associated with the quadratic form
$Q^{(2)}_{P,L}$. It is not immediately clear, however, that the functions
in $\cD(H)$ possess $H^2$-regularity. If this is the case we say, for short, 
that the quadratic form is {\it regular}. We note that a self-adjoint
realisations of $-\Delta_2$ with domain \eqref{DomainGen} would correspond to 
a regular form. 

Under an additional (mild) assumption a regular quadratic form indeed leads
to a two-particle Laplacian with domain \eqref{DomainGen}.
\begin{prop}
\label{SFGeneral}
Suppose that the map $P$ is of class $C^{1}$ and that the quadratic form 
$Q^{(2)}_{P,L}$ with domain $\cD_{Q^{(2)}}$ is regular. Then the unique, 
self-adjoint and semibounded operator that is associated with this form
is the two-particle Laplacian $-\Delta_2$ with domain $\cD_{2}(P,L)$.
\end{prop}
\begin{proof}
The proof can essentially be taken over verbatim from the corresponding proof 
in \cite{BolKer11}. It is based on the representation theorem for quadratic 
forms, which implies that for each $\Psi\in\cD(H)$ there exists a unique 
$\chi\in L^{2}(D^{*}_{\Gamma})$ such that
\begin{equation}
\label{Rep}
 Q^{(2)}_{P,L}[\Psi,\Phi]= \langle\chi,\Phi\rangle\ , \quad  
 \forall\phi\in\cD_{Q^{(2)}} \ .
\end{equation}
When $\Phi \in C^{\infty}_{0}(D^{*}_{\Gamma})$, an integration by parts of 
\eqref{Rep} implies that $H$ acts as a two-particle Laplacian $-\Delta_{2}$. 
In the general case of a $\Psi\in\cD_{Q^{(2)}}$ the integration by parts yields
an additional boundary term, 
\begin{equation}
 -\int_0^1\langle\Psi'_{bv}(y)+L(y)\Psi_{bv}(y),\Phi_{bv}(y)\rangle_{\kz^{n(E)}}
 \ \ud y\ ,
\end{equation}
that is required to vanish. Following Lemma~3.13 in \cite{BolKer11}, which
has an immediate generalisation to the present case, the set
$\{\Phi_{bv};\ \Phi\in\cD_{Q^{(2)}}\}$ is dense in 
$\ker \Pi \subset L^2(0,1)^{n(E)}$. Hence, 
$\Psi'_{bv}+\Lambda\Psi_{bv}\in\ker\Pi^{\perp}$, or
\begin{equation}
 Q(y)\Psi'_{bv}(y)+Q(y)L(y)\Psi_{bv}(y)=0\ .
\end{equation}
This condition finally implies that $\cD(H)=\cD_2(P,L)$.
\end{proof}
As mentioned above, the quadratic forms in Proposition~\ref{2quadformDelta}
are not necessarily regular. Since our focus is on contact interactions of
$\delta$- or hardcore-type, it is sufficient to consider these cases. 
However, as in \cite{BolKer11} we have to add one additional assumption on
the restrictions $P_{vert}$ of the projectors $P$ to $V_{vertex}$. Splitting
$V_{vertex}$ into the two subspaces spanned by the upper two and the lower
two components of \eqref{bvnondiag} as well as the upper two and middle two
components of \eqref{bvdiag}, respectively, we require $P_{vert}$
to be block-diagonal with respect to this decomposition.

This then leads to the main result of this section.
\begin{theorem}
\label{Regular}
In addition to the assumption made for the maps $P$ and $L$ above, suppose 
that $P$ is of class $C^{3}$ and $L$ is Lipschitz continuous. 
Furthermore, for $y\in[0,\epsilon_{1}]\cup[l-\epsilon_{2},l]$ with some 
$\epsilon_{1},\epsilon_{2} > 0$ assume that the restriction of $P$ to 
$V_{vertex}$ is diagonal with diagonal entries zero or one and, in the case of 
$\delta$-type interactions, that $\alpha(y)=\alpha_0\geq 0$ for those $y$. 
Then the quadratic form $Q^{(2)}_{P,L}$ is regular.
\end{theorem}
\begin{proof} 
First note that it is enough to show regularity near the corners of 
$D_{ee}=D^+_{ee}\cup D^-_{ee}$ adjacent to the diagonal. The regularity away from
the diagonal of $D_{ee}$ as well as regularity in the rectangles $D_{e_{1}e_{2}}$  
with $e_{1} \neq e_{2}$ was already established in \cite{BolKer11}. In addition, 
the regularity along the diagonals in the interior of $D_{ee}$ can be readily 
established using the same methods as in \cite{BolKer11}. 

The assumptions made on $P$ imply that on the edges of the squares $D_{ee}$
the functions in $\cD_{Q^{(2)}}$ satisfy either Dirichlet- or Neumann boundary 
conditions near the corners. Along diagonals we consider the projections
\begin{equation}
\label{diagBF}
\begin{split}
 \psi_{ee,B}(x,y):=&\frac{1}{2}\bigl[\psi_{ee}(x,y)+\psi_{ee}(y,x)\bigr]\ , \\
 \psi_{ee,F}(x,y):=&\frac{1}{2}\bigl[\psi_{ee}(x,y)-\psi_{ee}(y,x)\bigr]\ .
\end{split}
\end{equation}
The goal is to show that, close to the corners, both $\psi_{ee,B}$ and 
$\psi_{ee,F}$ are of class $H^{2}$. For that purpose one introduces suitable 
cut-offs that restrict the functions \eqref{diagBF} to neighbourhoods of the 
corners. This eventually implies that $\psi_{ee}\in H^{2}(D^\ast_{ee})$.

We recall the conditions \eqref{derjump} which imply that
\begin{equation}
 \partial_n\psi^\pm_{ee,B} - \frac{\alpha}{2\sqrt{2}}\psi^\pm_{ee,B} = 0 \ ,
\end{equation}
on the diagonal. Hence, $\psi^\pm_{ee,B}$ satisfies (variable) Robin boundary 
conditions on the diagonal. By construction, $\psi^\pm_{ee,F}$ vanishes on the 
diagonal so that near the corners of $D^\pm_{ee}$ adjacent to the diagonal, 
where $\alpha$ is supposed to be constant, $\psi^{\pm}_{ee,B/F}$ satisfies a 
combination of Dirichlet-, Neumann- or standard Robin-boundary conditions. 
In all such cases regularity is well known to hold \cite{Dau88,Gri85}.
\end{proof}
One naturally expects the two-particle operators representing contact 
interactions to possess purely discrete spectra of the form
$\lambda_0\leq\lambda_1\leq\lambda_2\leq\dots$ (i.e., eigenvalues are
counted with their multiplicities and do not accumulate at any finite value).
Moreover, the asymptotic distribution of eigenvalues, as given by the 
asymptotic behaviour of the eigenvalue-counting function
\begin{equation}
\label{evcount}
 N(\lambda) := \#\{n;\ \lambda_n\leq\lambda\}\ ,
\end{equation}
should follow a Weyl law. We shall now prove a Weyl law for repulsive contact 
interactions. This includes hardcore- and $\delta$-interactions with 
$\alpha\geq 0$. The general requirement is that $L_{contact}$ is negative 
definite (compare \eqref{BCdelta2}).
\begin{prop}
\label{SpectrumDelta}
Let $(-\Delta_{2},\cD_{2}(P,L))$ be a self-adjoint realisation of the 
two-particle Laplacian with repulsive contact interaction as described in 
Proposition~\ref{SFGeneral}. Then this operator has compact resolvent.
In particular, its spectrum is purely discrete and only accumulates at
infinity. Furthermore, the counting function \eqref{evcount} obeys the Weyl
law
\begin{equation}
\label{Weyl}
 N(\lambda) \sim \frac{\cL^2}{4\pi}\,\lambda\ ,\quad\lambda\to\infty\ ,
\end{equation}
where $\cL=\sum_{e=1}^{E}l_{e}$ is the total length of the graph. 
\end{prop}
\begin{proof}
The Hilbert space $H^{1}(D^{*}_{\Gamma})$ is compactly embedded in 
$L^{2}(D^{*}_{\Gamma})$. Since the form norm $||\cdot||_{Q^{(2)}}$ is equivalent
to the $H^{1}(D^{*}_{\Gamma})$-norm, the Hilbert space 
$(\cD_{Q^{(2)}},||\cdot||_{Q^{(2)}})$ is also compactly embedded in 
$L^{2}(D^{*}_{\Gamma})$. Hence the operator associated with the quadratic form
has compact resolvent \cite{Dob05}.

The Weyl law follows from a standard bracketing argument \cite{ReeSim78} based 
on a comparison with two suitable operators (quadratic forms), see also
\cite{BolEnd09,BolKer11}.

The first operator, $(-\Delta_2,\cD_{2}(P_D,L_D))$, is the 
Dirichlet-Laplacian, and is characterised by the projector $P_D=\eins$ as well 
as $L_D=0$. Given an operator $(-\Delta_{2},\cD_{2}(P,L))$ the second comparison
operator, $(-\Delta_2,\cD_{2}(P_R,L_R))$, is the Laplacian given by the 
projector $P_{R}=0$ as well as 
\begin{equation}
 L_R=\diag(\underbrace{\lambda,\dots,\lambda}_{4E^{2}-\text{times}},
 \underbrace{0,\dots,0}_{2E-\text{times}}) \ ,
\end{equation}
where $\lambda=||\Lambda||_{\op}$. The associated quadratic forms therefore 
satisfy the following inclusions of their domains,
\begin{equation}
 \cD_2(P_D,L_D) \subseteq \cD_2(P,L) \subseteq \cD_2(P_{R},L_{R})\ .
\end{equation}
Hence \cite{ReeSim78}, it follows that the related eigenvalue-counting
functions satisfy
\begin{equation}
\label{Weylbounds}
 N_D(\lambda)\leq N(\lambda)\leq N_R(\lambda)\ .
\end{equation}
As both $N_D$ and $N_R$ satisfy the Weyl law \eqref{Weyl}, the same asymptotics
hold for $N(\lambda)$. 
\end{proof}

\section{Contact interactions for bosons}
\label{3sec}
So far we assumed to have two non-identical particles on a graph. We now
implement an exchange symmetry for two identical bosonic particles. Their 
states are described in the symmetric two-particle Hilbert-space 
$\cH_{2,B}=\cH_1\otimes_s\cH_1$. The orthogonal projection $\Pi_B$ from 
$\cH_2=L^2(D^\ast_\Gamma)$ to the bosonic subspace 
\begin{equation}
 L^2_B(D^\ast_\Gamma)=\Pi_B L^2(D^\ast_\Gamma)
\end{equation}
acts on components of
$\Psi=(\psi_{e_1e_2})$ with $e_1\neq e_2$ as
\begin{equation}
 (\Pi_{B}\Psi)_{e_{1}e_{2}}(x_{e_{1}},y_{e_{2}})=
 \frac{1}{2}\bigl(\psi_{e_{1}e_{2}}(x_{e_{1}},y_{e_{2}})+
 \psi_{e_{2}e_{1}}(y_{e_{2}},x_{e_{1}})\bigr)\ ,
\end{equation}
whereas on the components with $e_1=e_2$ its action reads
\begin{equation}
 (\Pi_{B}\Psi)^\pm_{ee}(x_e,y_e) = \frac{1}{2}\bigl(\psi^\pm_{ee}(x_e,y_e)+
 \psi^\mp_{ee}(y_e,x_e)\bigr)\ .
\end{equation}
Due to this symmetry it would be sufficient to keep only components with
$e_1<e_2$ in addition to the diagonal components with $e_1=e_2$. For
simplicity, when comparing to the previous section we, however, keep all 
components. We then denote the images of function spaces under the projection 
to their bosonic subspaces as, e.g.,
\begin{equation}
 H^m_B(D^\ast_\Gamma)=H^m(D^\ast_\Gamma)\cap\cH_{2,B}\ .
\end{equation}
We note that whenever $\Psi\in\cH_{2,B}$ is in $H^2(D^\ast_{\Gamma})$, the 
underlying symmetry implies the relations
\begin{equation}
 \psi_{e_1e_2,x}(x_{e_1},y_{e_2})=\psi_{e_2e_1,y}(y_{e_2},x_{e_1})
 \quad\text{and}\quad
 \psi_{e_1e_2,xx}(x_{e_1},y_{e_2})=\psi_{e_2e_1,yy}(y_{e_2},x_{e_1})\ ,
\end{equation}
when $e_1\neq e_2$, as well as
\begin{equation}
 \psi^\pm_{ee,x}(x_e,y_e)=\psi^\mp_{ee,y}(y_e,x_e)
 \quad\text{and}\quad
 \psi^\pm_{ee,xx}(x_e,y_e)=\psi^\mp_{ee,yy}(y_e,x_e)\ .
\end{equation}
Due to the bosonic symmetry it is possible to reduce the number of
components in the vectors of boundary values \eqref{bvdef}. When 
$e_1\neq e_2$, it suffices to keep the upper two components in each of
the vectors \eqref{bvnondiag}, whereas for $e_1=e_2$ we use
\begin{equation}
\label{BVGeneralBosons}
 \psi_{ee,bv}(y) = \begin{pmatrix}  
 \sqrt{l_e}\psi^{-}_{ee}(0,l_e y) \\ \sqrt{l_e}\psi^{+}_{ee}(l_e,l_e y) \\ 
 \sqrt{l_e}\psi^{+}_{ee}(l_e y,l_e y) 
 \end{pmatrix}\qquad\text{and}\qquad
 \psi^{'}_{ee,bv}(y) = \begin{pmatrix} 
 \sqrt{l_e}\psi^{-}_{ee,x}(0,l_e y) \\ -\sqrt{l_e}\psi^{+}_{ee,x}(l_e,l_e y) \\ 
 \sqrt{2l_e}\psi^{+}_{ee,n}(l_e y,l_e y)
 \end{pmatrix}\ ,
\end{equation}
with $y\in[0,1]$. The space of boundary values therefore has dimension 
$n_{B}(E)=2E^{2}+E$, and decomposes in analogy to \eqref{bvblocks}.

We also need the bounded and measurable maps $P,L:[0,1]\to\M(n_{B}(E),\kz)$, 
where
\begin{enumerate}
\item $P(y)$ is an orthogonal projector,
\item $L(y)$ is a self-adjoint endomorphism on $\ker P(y)$,
\end{enumerate}
for a.e.\ $y \in [0,1]$. The space $\kz^{n_B(E)}$ of boundary values decomposes
in the same way as \eqref{bvblocks}, however, the edge-$e$ subspaces are now
one-dimensional. This forces the equivalent of \eqref{BCdelta1} to be a 
projector on $\kz$ and to take values
\begin{equation}
\label{BCdelta1_B}
 P_{contact,e}(y) \in \{0,1\}\ .
\end{equation}
Here $P_{contact,e}(y)=1$ corresponds to a Dirichlet condition in the point 
$(l_e y,l_e y)$ along the diagonal, whereas $P_{contact,e}(y)=0$ imposes no 
condition. Hence, when $\delta$-type interactions are considered we choose 
$P_{contact,e}(y)=0$, and in the case of hardcore-interactions $P_{contact,e}(y)=1$
is chosen. Likewise, the equivalent of \eqref{BCdelta2} is
\begin{equation}
\label{BCdeltaB}
 L_{contact,e}(y) = -\frac{1}{2}\,\alpha(y) 
\end{equation}
for $\delta$-interactions, and $L_{contact,e}(y) =0$ for interactions of hardcore
type.

We can now set up the following quadratic form,
\begin{equation}
\label{QformDeltaGeneralBosons}
 Q^{(2),B}_{P,L}[\psi] = 2\langle\Psi_x,\Psi_x\rangle_{L^2_B(D^\ast_\Gamma)}
-2\int_0^1 \langle\Psi_{bv}(y),L(y)\Psi_{bv}(y)\rangle_{\kz^{n_B(E)}}\ \ud y \ ,
\end{equation}
with domain
\begin{equation}
\label{DefquadDeltaBosons}
 \cD_{Q^{(2),B}} = \{ \Psi\in H^1_B(D^\ast_\Gamma);\ P(y)\Psi_{bv}(y)=0\ 
 \text{for a.e.}\ y \in [0,1]\}\ .
\end{equation}
As this is the restriction of a quadratic form on $L^2(D^\ast_\Gamma)$ to
$L^2_B(D^\ast_\Gamma)$, all results of Section~\ref{2sec} carry over: 
Propositions~\ref{2quadformDelta} and \ref{SFGeneral} imply that the 
quadratic form is closed and semi-bounded; when $P$ is of class $C^1$ and 
the form is regular, the associated self-adjoint operator is the bosonic 
two-particle Laplacian $-\Delta_{2,B}$ with domain
\begin{equation}
\begin{split}
 \cD_{2,B}(P,L) := \{
       &\Psi\in H^2_B(D^\ast_\Gamma);\ P(y)\Psi_{bv}(y)=0\ \text{and}\\
       &\quad Q(y)\Psi'_{bv}(y)+L(y)Q(y)\Psi_{bv}(y)=0\ \text{for a.e.}\
          y\in [0,1] \} \ .
\end{split}
\end{equation}
According to Theorem~\ref{Regular}, when the consitions of that theorem are 
fulfiled the quadratic forms leading to $\delta$-type and hardcore-interactions
are regular. 

We remark that for $\delta$-interactions one can use the decomposition
\eqref{bvblocks} of the space of boundary values and the explicit expression
\eqref{BCdeltaB} to rewrite the quadratic form as
\begin{equation}
\label{QformDeltaBalt}
\begin{split}
 Q^{(2),B}_{P,L}[\psi] 
   &= 2\langle\Psi_x,\Psi_x\rangle_{L^2_B(D^\ast_\Gamma)} 
      -2\int_0^1 \langle\Psi_{bv,vert}(y),L_{vert}(y)\Psi_{bv,vert}(y)
           \rangle_{\kz^{2E^2}}\ \ud y \\
   &\quad + \sum_{e=1}^E\int_0^1\alpha(y)\,|\sqrt{l_e}
          \psi_{ee}^+(l_e y,l_e y)|^2\ \ud y \ .
\end{split}
\end{equation}
In the same way, the form domain takes the form
\begin{equation}
\label{DefquadDeltaBalt}
 \cD_{Q^{(2),B}} = \{ \Psi\in H^1_B(D^\ast_\Gamma);\ P_{vert}(y)\Psi_{bv,vert}(y)=0\ 
 \text{for a.e.}\ y \in [0,1]\}\ .
\end{equation}

Due to the bosonic projection $\Pi_B$, which commutes with any of the 
two-particle Laplacians, asymptotically half of the spectrum of a Laplacian is 
projected to the bosonic Hilbert space $\cH_{2,B}$, so that 
Proposition~\ref{SpectrumDelta} implies the Weyl law
\begin{equation}
\label{WeylB}
 N_B(\lambda) \sim \frac{\cL^2}{8\pi}\,\lambda\ ,\quad\lambda\to\infty\ ,
\end{equation}
for the asymptotics of the eigenvalue count restricted to $\cH_{2,B}$.

Our goal now is to study bosonic many-particle systems on graphs. Eventually,
these have to be described in the bosonic Fock space over the one-particle
Hilbert space. Since it suffices, however, to consider each $N$-particle
space separately, we here only consider a fixed particle number $N$. In that
context we shall introduce two-particle interactions that are (formally) of 
the type,
\begin{equation}
\label{formalintHam}
 H_N = -\Delta_N +\sum_{i<j}\alpha(x_i)\,\delta(x_i-x_j)\ .
\end{equation}
Due to the bosonic symmetry, on suitable functions the quadratic form 
associated with such an operator will be
\begin{equation}
\label{Nbosquad}
\begin{split}
 &\langle\Psi,H_N\Psi\rangle_{\cH_{2,B}} = 
      \langle\Psi,-\Delta_N\Psi\rangle_{\cH_{2,B}} \\
 &\qquad+ \frac{N(N-1)}{2}\sum_{e_2\dots e_N}\int_0^{l_{e_2}}\dots\int_0^{l_{e_N}}
      \alpha(x_{e_2})|\psi_{e_2 e_2\dots e_N}(x_{e_2},x_{e_2},\dots,x_{e_N})|^2
      \ \ud x_{e_N}\dots\ud x_{e_2}\ .
\end{split}
\end{equation}
From \eqref{formalintHam} and \eqref{Nbosquad} one concludes that contact 
interactions involve boundary values along hypersurfaces that are characterised
by the fact that two particles are at the same position. 

The configuration space of $N$ (distinguishable) particles is
\begin{equation}
 D^N_\Gamma = \bigcup_{e_1e_2...e_{N}} D_{e_1e_2...e_{N}}\ ,
\end{equation}
where $D_{e_1e_2...e_{N}}=(0,l_{e_1})\times\dots\times (0,l_{e_N})$. We stress that
this notation includes cases where several particles are on the same edge, in
which case the same edge appears repeatedly. The hyperplanes that determine 
contact interactions are characterised by equations $x^i_e = x^j_e$, meaning 
that particles $i$ and $j$ sit on the same position on edge $e$. In analogy to 
\eqref{Ddecomp}, in order to implement contact interactions we have to 
decompose $D^N_\Gamma$ further; this involves all hyperrectangles 
$D_{e_1e_2...e_{N}}$ that are composed of at least two coinciding edges. 

Now assume that $(n_1,\dots,n_E)$ is a partition of $N$ such that there are 
$n_{e}$ particles on edge $e$. Let $\sigma\in S_N$ assigns labels to the $N$ 
particles in such a way that $\sigma(1),\dots,\sigma(n_e)$ label the particles 
on edge $e$, with coordinates $x^{\sigma(1)}_{e},\dots,x^{\sigma(n_{e})}_{e}$. 
Permutations $\pi\in S_{n_1}\times\dots\times S_{n_E}\subset S_N$ of particle 
labels then leave the assignment to edges untouched, and there exists such a 
permutation with
\begin{equation}
 x^{\pi(\sigma(1))}_{e} \leq\dots\leq x^{\pi(\sigma(n_{e}))}_{e}\ ,\quad\forall e\in
 \{1,\dots,E\}\ .
\end{equation}
These relations define a polyhedral subdomain of $D_{e_{1}...e_{N}}$. Every other
permutation $\pi\in S_{n_1}\times\dots\times S_{n_E}$ will produce a copy of
that polyhedral subdomain that emerges through reflections in a succession of  
boundary hyperplanes. We will enumerate these $n_{1}!...n_{E}!$ subdomains as 
$D_{e_{1}...e_{N}}^\eta$, with $1\leq\eta\leq n_{1}!...n_{E}!$. 

In analogy to \eqref{Ddecomp} we now introduce the dissected hyperrectangles
as the disjoint union
\begin{equation}
 D^{\ast}_{e_{1}...e_{N}}=\dot{\bigcup_{\eta}}\ D_{e_{1}...e_{N}}^\eta \ .
\end{equation}
The $N$-particle Hilbert space $\cH_{N}$ for $N$ (distinguishable) particles 
with contact interactions can then be defined as
\begin{equation}
\label{HilbertSpaceContactI}
 L^2(D^{N\ast}_\Gamma)= \bigoplus_{e_1e_2...e_{N}} L^2(D^{*}_{e_1e_2...e_{N}}) \ .
\end{equation}
The corresponding Sobolev spaces are defined in the same way. Note that on
the right-hand side $D^{*}_{e_1e_2...e_{N}}=D_{e_1e_2...e_{N}}$ when all edges in
the definition of $D_{e_1e_2...e_{N}}$ are distinct, i.e., when no two particles
are on the same edge. With this proviso, we denote functions on
$D^{*}_{e_1e_2...e_{N}}$ by $\psi_{e_{1}...e_{N}}$, which themselves consist of 
$n_{1}!...n_{E}!$ components defined on the subdomains $D_{e_{1}...e_{N}}^\eta$,
compare \eqref{Ddecomp} and below.

The projection $\Pi_B$ from \eqref{HilbertSpaceContactI} to the bosonic 
$N$-particle Hilbert space $L^2_{B}(D^{N\ast}_\Gamma)$ then is given by
\begin{equation}
 (\Pi_{B}\Psi)_{e_{1}...e_{N}} = 
 \frac{1}{N!}\sum_{\pi \in S_{N}} \psi_{\pi(e_{1})...\pi(e_{N})}
 (x^{\pi(1)}_{\pi(e_{1})},...,x^{\pi(N)}_{\pi(e_{N})}) \ .
\end{equation}

In analogy to \eqref{fctcont} and \eqref{derjump}, two-particle interactions 
of a $\delta$-type \eqref{formalintHam} require boundary values of functions
$\Psi \in H^1_{B}(D^{N\ast}_\Gamma)$ and their normal derivatives along (internal)
boundary hyperplanes of the dissected hyperrectangles $D^{*}_{e_1e_2...e_{N}}$. In 
addition, boundary conditions at vertices have to be implemented. For those 
purposes the most convenient expression for the quadratic form is an
analogue of \eqref{QformDeltaBalt}.

We first introduce the vectors of boundary values in vertices. Due to the 
bosonic symmetry these can be given in the form
\begin{equation}
\label{BVI}
 \Psi_{bv,vert}(\vy) = \begin{pmatrix} \sqrt{l_{e_{2}}\dots l_{e_{N}}} 
 \psi_{e_{1}\dots e_{N}}(0,l_{e_{2}}y_{1},\dots,l_{e_{N}}y_{N-1}) \\ 
 \sqrt{l_{e_{2}}\dots l_{e_{N}}} 
 \psi_{e_{1}\dots e_{N}}(l_{e_{1}},l_{e_{2}}y_{1},\dots,l_{e_{N}}y_{N-1}) \end{pmatrix} \ ,
\end{equation}
and
\begin{equation}
\label{BVII}
 \Psi^{'}_{bv,vert}(\vy) = \begin{pmatrix} \sqrt{l_{e_{2}}\dots l_{e_{N}}}  
 \psi_{e_{1}\dots e_{N},x^{1}_{e_{1}}}(0,l_{e_{2}}y_{1},\dots,l_{e_{N}}y_{N-1}) \\ 
 -\sqrt{l_{e_{2}}\dots l_{e_{N}}} 
 \psi_{e_{1}\dots e_{N},x^{1}_{e_{1}}}(l_{e_{1}},l_{e_{2}}y_{1},\dots,l_{e_{N}}y_{N-1})  
 \end{pmatrix} \ ,
\end{equation}
where $\vy=(y_{1},\dots,y_{N-1})\in [0,1]^{N-1}$. On these (vertex related)
boundary values the bounded and measurable maps 
$P_{vert},L_{vert}: [0,1]^{N-1} \to \M(2E^{N},\kz)$ shall act, which are required 
to fulfil
\begin{enumerate}
\item $P_{vert}(\vy)$ is an orthogonal projector,
\item $L_{vert}(\vy)$ is a self-adjoint endomorphism on $\ker P_{vert}(\vy)$,
\end{enumerate}
for a.e. $\vy\in [0,1]^{N-1}$. 

Boundary values on internal hyperplanes in the dissected hyperrectangles
involve components $\psi^\eta_{e_1\dots e_N}$ with a pair of coinciding edges,
$e_i =e_j$. Due to the exchange symmetry we can always arrange for these
edges to be $e_1$ and $e_2=e_1$. Permuting a given pair $(e_i,e_j)$ to
$(e_1,e_2)$, however, involves a change of the associated domain 
$D^\eta_{e_1\dots e_N}$ to some other copy $D^{\eta'}_{e'_1\dots e'_N}$. This means
that $\psi^\eta_{e_1\dots e_N}$ is replaced by $\psi^{\eta'}_{e'_1 e'_1\dots e'_N}$.

In analogy to \eqref{QformDeltaBalt} the quadratic form we wish to set up is
\begin{equation}
\label{QuadFormContact}
\begin{split}
 Q^{(N)}_{B}[\Psi] 
  &= N \sum_{e_{1}\dots e_{N}}\int_{0}^{l_{e_{1}}}\dots\int_{0}^{l_{e_{N}}} 
    |\psi_{e_{1}\dots e_{N},x_{e_1}}(x_{e_{1}},\dots,x_{e_{N}})|^{2}\ \ud x_{e_N}\dots
    \ud x_{e_1} \\
  &\quad -N\int_{[0,1]^{N-1}}\langle\Psi_{bv,vert},L_{vert}(\vy)\Psi_{bv,vert} 
    \rangle_{\kz^{2E^{N}}} \ud\vy \\
  &\quad +\frac{N(N-1)}{2}\sum_{e_{2}...e_{N}}\int_{[0,1]^{N-1}} \alpha(y_1)\ 
    |\sqrt{l_{e_{2}}\dots l_{e_{N}}}\psi_{e_{2}e_{2}\dots e_{N}}
    (l_{e_2}y_1,\vl\vy)|^{2}\ d\vy \ .
\end{split}
\end{equation}
For convenience we here used the notation 
$\vl\vy=(l_{e_2}y_1,l_{e_3}y_2,\dots,l_{e_N}y_{N-1})$.

This form shall be defined on the domain
\begin{equation}
\begin{split}
 \cD_{Q^{(N)}_{B}} = \{\Psi \in H^{1}_{B}(D^{N\ast}_\Gamma);\ P_{vert}(\vy)
 \Psi_{bv,vert}(\vy)=0\ \text{for a.e.}\ \vy\in [0,1]^{N-1}\}\ .
\end{split}
\end{equation}
Using this, we can readily establish the following statements. These are 
immediate generalisations of the corresponding statements, 
Propositions~\ref{2quadformDelta} and \ref{SFGeneral}, for two bosons ($N=2$).
\begin{theorem}
Let the maps $P_{vert},L_{vert}:[0,1]^{N-1}\to \M(2E^{N},\kz)$ as well as the 
function $\alpha:[0,1]\to\kz$ be bounded and measurable. Then:
\begin{itemize}
\item[(i)] The quadratic form $Q^{(N)}_{B}$ defined on the domain 
$\cD_{Q^{(N)}_{B}}$ is closed and semi-bounded.
\item[(ii)] If $P_{vert}$ is of class $C^1$ and the form is regular, the 
associated self-adjoint operator is the $N$-particle Laplacian $-\Delta_N$
with domain
\begin{equation}
\begin{split}
 \cD_{N,B}(P,L) := \{
   &\Psi\in H^2_B(D^{N\ast}_\Gamma);\ P(\vy)\Psi_{bv}(\vy)=0\ \text{and}\\
   &\quad Q(\vy)\Psi'_{bv}(\vy)+L(\vy)Q(\vy)\Psi_{bv}(\vy)=0\ \text{for a.e.}\
          \vy\in [0,1]^{N-1} \} \ .
\end{split}
\end{equation}
Here $P=P_{contact}\oplus P_{vert}$ and $L=L_{contact}\oplus L_{vert}$ refer to all 
boundary values.
%
\end{itemize}
\end{theorem}
The proof of this Theorem is an immediate extension of the proofs of
Propositions~\ref{2quadformDelta} and \ref{SFGeneral} as well as of 
Theorem~\ref{Regular}.

It is also immediately clear from the proof of Proposition~\ref{SpectrumDelta}
that any of the $N$-particle Laplacians $-\Delta_N$ with repulsive contact
interactions have compact resolvent
and hence possess purely discrete spectra, accumulating only at infinity.
Furthermore, the eigenvalue counting function (compare \eqref{evcount}
and \eqref{WeylB}) satisfies a Weyl law that follows from a bracketing 
argument in the same way as \eqref{Weyl}. For the case of $N$ distuingishable
particles the Weyl law is
\begin{equation}
 \label{WeylBN}
 N(\lambda) \sim \frac{\cL^N}{(4\pi)^{N/2}\Gamma(1+\frac{N}{2})}\,
 \lambda^{N/2}\ ,\quad\lambda\to\infty\ .
\end{equation}
This follows most easily from the lower bound given by the Dirichlet Laplacian
as in \eqref{Weylbounds}. The bosonic case requires to desymmetrise the 
spectrum with respect to particle exchange symmetry; hence, the bosonic
counting function is reduced by a factor of $\frac{1}{N!}$.

\vspace*{0.5cm}

\subsection*{Acknowledgement}
J K would like to thank the {\it Evangelisches Studienwerk Villigst} for
financial support through a Promotionsstipendium.

\vspace*{0.5cm}

{\small
\bibliographystyle{amsalpha}
\bibliography{literatur}}

\end{document}